\documentclass[10pt,letter, journal]{IEEEtran}
\usepackage[centertags]{amsmath}
\usepackage{amsfonts}
\usepackage{amssymb}
\usepackage{amsmath,  amssymb}
\usepackage{graphicx}
\usepackage[centertags]{amsmath}
 \DeclareMathOperator{\rank}{rank}

\begin{document}
\title{On the Index Coding Problem and its Relation to Network Coding and Matroid Theory }


\author{\authorblockN{Salim Y. El Rouayheb, Alex Sprintson, and Costas N. Georghiades}\\
\authorblockA{
Department of Electrical and Computer Engineering\\
Texas A\&M University\\
College Station, TX, USA\\
Email: \{salim, spalex, c-georghiades\}@ece.tamu.edu}
}
\maketitle

\newtheorem{theorem}{Theorem}
\newtheorem{lemma}[theorem]{Lemma}
\newtheorem{claim}[theorem]{Claim}
\newtheorem{proposition}[theorem]{Proposition}
\newtheorem{condition}[theorem]{Condition}
\newtheorem{observation}[theorem]{Observation}
\newtheorem{conjecture}[theorem]{Conjecture}
\newtheorem{property}[theorem]{Property}
\newtheorem{assertion}[theorem]{Assertion}
\newtheorem{example}[theorem]{Example}
\newtheorem{definition}[theorem]{Definition}
\newtheorem{corollary}[theorem]{Corollary}
\newtheorem{remark}[theorem]{Remark}
\newtheorem{note}[theorem]{Note}
\newtheorem{problem}[theorem]{Problem}

\begin{abstract}

The \emph{index coding} problem has recently attracted a significant attention from the research community due to its theoretical significance and applications in wireless ad-hoc networks. An instance of the index coding problem includes a sender that holds a set of information messages $X=\{x_1,\dots,x_k\}$ and a set of receivers $R$. Each receiver $\rho=(x,H)\in R$ needs to obtain a message $x\in X$ and has prior \emph{side information} comprising a subset $H$ of  $X$. The sender uses a noiseless communication channel to broadcast encoding of messages in $X$ to all clients. The objective is to find an encoding scheme that minimizes the number of transmissions required to satisfy the receivers' demands with \emph{zero error}.

In this paper, we analyze the relation between the index coding problem, the more general network coding problem and the problem of finding a linear representation of a matroid. In particular, we show that any instance of the network coding and matroid representation problems can be efficiently reduced to an instance of the index coding problem. Our reduction implies that many important properties of the network coding and matroid representation problems carry over to the index coding problem. Specifically, we show that \emph{vector linear codes} outperform scalar linear codes and that vector linear codes are insufficient for achieving the optimum number of transmissions.

\end{abstract}

\section{Introduction}

In recent years there has been a significant interest in utilizing the broadcast nature of wireless signals to improve the throughput and reliability of ad-hoc wireless networks. The wireless medium allows the sender node to deliver data to several neighboring nodes with a single transmission. Moreover, a wireless node can opportunistically listen to the wireless channel and store all the obtained packets, including those designated for different nodes. As a result, the wireless nodes can obtain side information which, in combination with  proper encoding techniques, can lead to a substantial improvement in the performance of the wireless network.

Several recent studies focused on wireless architectures that utilize the broadcast properties of the wireless channel by using  coding techniques. In particular, \cite{KRHKMC06,Katti_SymbolLevelNetworkCodingForWirelessMeshNetworks_2008} proposed new architectures, referred to as COPE and MIXIT, in which  routers mix packets from different information sources to increase the overall network throughput. Birk and Kol  \cite{BK98,BK06} discussed applications of coding techniques in satellite networks with caching clients with a low-capacity reverse channel \cite{BK98,BK06}.

The major challenge in the design of opportunistic wireless networks is to identify an optimal encoding scheme that minimizes the number of transmissions necessary to satisfy all client nodes. This can be formulated as  the  \emph{Index Coding} problem that includes a single sender node $s$ and a set of receiver nodes $R$. The sender has a set of information messages $X=\{x_1,\dots,x_k\}$ that need to be delivered to the receiver nodes. Each receiver $\rho=(x,H)\in R$ needs to obtain a single message  $x$ in $X$ and has prior \emph{side information} comprising a subset $H\subseteq X$. The sender can broadcast the encoding of messages in $X$ to the receivers through a noiseless channel that has a capacity of one message per channel use. The objective is to find an optimal encoding scheme, referred to as an \emph{index code}, that satisfies all receiver nodes with the minimum number of transmissions.

\begin{figure}[t]
\begin{center}
\includegraphics[width=3in]{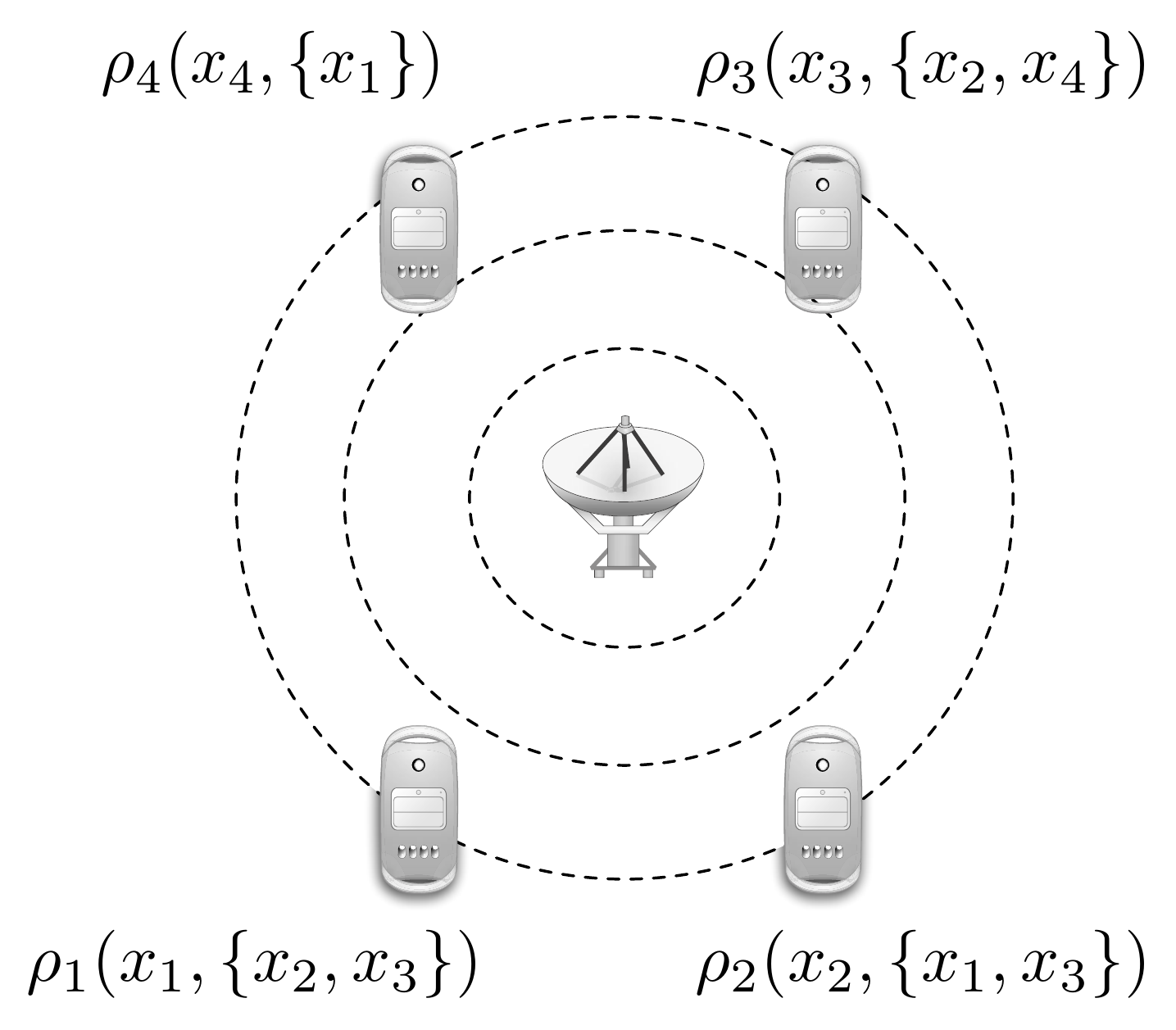}
 \caption{An instance of the index coding problem with four messages and four clients. Each client is represented by a couple (x,H), where $x\in X$ is the packet demanded by the client, and $H\subseteq X$ represent its side information. \label{fig:butterfly}}%
\end{center}
\end{figure}

With a \emph{linear encoding scheme}, all messages in $X$ are elements of a finite field and all encoding operations are linear over that field. Figure \ref{fig:butterfly} depicts an instance of  the index coding problem that includes a sender with four messages $x_1,\dots, x_4$ and four clients. We assume that each message is an element of $GF(2^n)$, represented by $n$ bits. 
Note that the sender can satisfy the demands of all clients, in a straightforward way, by broadcasting all four messages over the wireless channel. The encoding operation achieves a reduction of the number of messages by a factor of two. Indeed, it is sufficient to send just two messages $x_1+x_2+x_3$ and $x_1+x_4$ (all operations are over $GF(2^n)$) to satisfy the requests of all clients. This example demonstrates that by using an efficient encoding scheme, the sender can significantly reduce the number of transmissions which, in turn, results in a reduction in delay and energy consumption.

The above example utilizes a \emph{scalar linear} encoding scheme that performs coding over the original messages. In a \emph{vector} encoding scheme, each message is divided into a number of smaller size messages, referred to as \emph{packets}. The vector encoding scheme combines packets from different messages to minimize the number of transmissions.  With \emph{vector linear} index coding, all packets are elements of a certain finite field $\mathbb{F}$, and each transmitted packet is a linear combination of the original packets. For example, consider the instance depicted in Figure \ref{fig:butterfly}, and suppose that each message $x_i,1\leq i\leq 4$ is divided into two packets,  $x_i^1, x_i^2 \in GF(2^{\bar{n}})$,  of size $\bar{n}=\frac{n}{2}$. Then, a valid vector-linear solution is comprised of four packets $\{x_1^1+x_4^1, x_1^2+x_4^2, x_1^1+x_2^1+x_3^1,  x_1^2+x_2^2+x_3^2\}$.


\subsection*{Related Work}

Witsenhausen \cite{W76}  considered a related zero-error coding problem with side information. He studied the point-to-point scenario where a transmitter wants to send a random variable $X$ over a noisy channel to a receiver that has prior knowledge of another  random variable $Y$ jointly distributed with $X$. The objective of this problem is to find an encoding scheme that allows the receiver to obtain $X$ with zero error probability. He observed that if the transmitter knows the realizations of $Y$, then the solution can be found in a straightforward way. However,  in the case when the transmitter is oblivious to the realizations of $Y$, then the problem becomes much harder. Specifically, the minimum communication rate in this case is related to the chromatic number of the powers of the channel confusion graph.    Simonyi \cite{S03} considered a more general case with multiple receivers and showed that the obvious lower bound, given by the rate needed by the ``weakest'' receiver, is  attainable. Index Coding can  be considered as a multi-terminal generalization of Witsenhausen's problem, with non-broadcast demands and with the restriction that the side information be a subset of the original one.

The index coding problem has been introduced by Birk and Kol \cite{BBJK06} and was initially motivated by broadcast satellite applications\footnote{Reference \cite{BBJK06} refers to the index coding problem as Informed Source Coding on Demand problem (ISCOD).}. In particular, they developed several heuristic solutions for this problem and proposed protocols for practical implementation in satellite networks.  Bar-Yossef et al.\,  studied the index coding problem from a graph-theoretical perspective \cite{BBJK06}. They showed that the index coding problem is equivalent to finding an algebraic property referred to as the \emph{minrank} of a graph. Finding, the minrank of a graph, however, is an intractable problem \cite{P96}. Lubetzky and Stav \cite{LS07} showed that non-linear scalar codes have a significant advantage over  linear ones by constructing a family of instances with an increasing gap between the optimal number of transmissions required by non-linear and linear codes. Wu et al. \cite{WPCPC061} studied the information-theoretic aspects of the problem with the goal of characterizing the admissible rate region\footnote{Reference \cite{WPCPC061} refers to the Index Coding problem as the Local Mixing Problem.}. Reference \cite{LS08} analyzed the hardness of approximation of the Index Coding problem.  References \cite{RCS07} and \cite{CS08} presented several heuristic solutions based on graph coloring and SAT solvers.

Index Coding can be considered as a special case of the Network Coding problem \cite{ACLY00}. The network coding technique extends the capability of intermediate network nodes by allowing them to mix packets received from different incoming edges. The goal of the network coding problem is to find the maximum rate between the source and destination pairs in a communication network with limited edge capacities. Initial works on the network coding technique focused on establishing \emph{multicast} connections. It was shown in \cite{ACLY00} and \cite{LYC03} that the capacity of a multicast network, i.e., the maximum number of packets that can be sent from the source $s$ to a set $T$ of terminals per time unit, is equal to the minimum capacity of all the  cuts that separates the source $s$ from any terminal $t\in T$.  In a subsequent work, Koetter and M\'{e}dard \cite{KM03} developed an algebraic framework for network coding and investigated linear network codes for directed graphs with cycles. This framework was used by Ho et al. \cite{HKMKE03} to show that linear network codes can be efficiently constructed through a randomized algorithm. Jaggi et al. \cite{JSCEEJT04} proposed a deterministic polynomial-time algorithm for finding feasible network codes in multicast networks.  References \cite{FS07, YLC06} provide a comprehensive overview of network coding.



\subsection*{Contributions}

In this paper, we study the relation between the index coding problem and the more general network coding problem. In particular, we establish a reduction that maps any instance of the network coding problem to a corresponding instance  of the index coding problem. We show that several important properties of the network coding problem carry over to  index coding. Specifically, by applying our reduction to the network presented in \cite{DFZ05}, we show that vector linear solutions are suboptimal. We also present an instance of the index coding problem  in which  splitting a message into two packets  yields a smaller number of transmissions than a scalar linear solution.


We also study the relation between the index coding problem and  matroid theory. In particular, we present a reduction that maps any matroid to an instance of the index coding problem such that the problem has a special optimal vector linear code that we call \emph{perfect index code} if and only if the matroid has  a multilinear representation. This construction constitutes a means to apply numerous  results in the rich field of matroid theory to index coding, and, in turn, to the network coding problem. Using results on the non-Pappus matroid, we give another example where vector linear codes outperform scalar linear codes.

The rest of the paper is organized as follows. In Section~\ref{sec:model}, we discuss our model and formulate the index and network coding problems. In Section~\ref{sec:Indexoid}, we present a reduction from the network coding problem to the index coding problem. In Section~\ref{sec:Indexoid} we discuss the relation between the index coding problem and matroid theory. In Section~\ref{sec:applications}, we apply our reductions to show the sub-optimality of linear and scalar index codes.
Next,  in Section~\ref{sec:discussion} we discuss the relationship between networks and matroids and introduce
a new  family of networks  with interesting properties. 
Finally, conclusions appear in Section~\ref{sec:conclusion}.

\section{Model}\label{sec:model}
In this section, we present a formulation of the network coding and index coding problems.


\subsection{Index Coding}
An instance of the index coding problem $\mathcal{I}(X,R)$ includes
\begin{enumerate}
  \item A set of $k$ messages  $X=\{x_1,\dots,x_k\}$,
  \item A set of clients or receivers $R \subseteq\{(x,H); x\in X, H\subseteq  X\setminus\{x\} \}$.
\end{enumerate}
Here, $X$ represents the set of messages available at the sender. Each message $x_i$ belong to a certain alphabet $\Sigma^n$. A client is represented by a pair $(x,H)$, where $x\in X$ is the message required by the client, and $H\subseteq X$ is set of messages available to the client as side information.
Note that in our model each client requests exactly one message. This does not incur any loss of generality as any client that requests multiple messages can be substituted by several clients that require a single different message and have the same side information as the original one.

 Each message $x_i$ can be divided into $n$ packets, and we write $x=(x_{i1},\dots,x_{in}) \in \Sigma^n$. We denote by  $\xi=(x_{11},\dots,x_{1n},\dots,x_{k1},\dots,x_{kn})\in \Sigma^{nk}$.



\begin{definition}[Index Code]\label{def:indexcoding}
 An  $(n,q)$ index code for $\mathcal{I}(X,R)$ is a function
\mbox{$f:\Sigma^{nk}\longrightarrow \Sigma^c$}, for a certain integer $c$, satisfying that for each client $\rho=(x,H)\in R$, there exists a function
    \mbox{$\psi_{\rho}: \Sigma^{c+n|H|}\longrightarrow \Sigma^n$} such that
    \mbox{$\psi_{\rho}(f(\xi),(x_i)_{x_i\in H})=x$}, $\forall \xi\in \Sigma^{nk}$ .
\end{definition}

We refer to $c$ as the \emph{length} of the index code.  Define $\ell(n,q)$ to be  the smallest integer $c$ such that the above condition holds for the given alphabet size $q$-ary  and block length $n$. If the index code satisfies $c=\ell(n,q)$, it is said to be \emph{optimal}.


We refer to  $\psi_\rho$ as the decoding function for client $\rho$. With a linear index code, the alphabet $\Sigma$ is a field and the functions $f$ and $\psi_\rho$
are linear in  variables $x_{ij}$. If $n=1$ the index code is called a scalar code, and for $n>1$, it is called a vector or block code. Note that in our model a packet can be requested by several clients. This is a slightly more general model than that considered in references \cite{BBJK06} and \cite{LS07} where it was assumed that each message can only be requested by a single client.


Given $n$ and $q$, the index coding problem consists of finding an optimal index code for an index coding instance. for a given    instance $\mathcal{I}(X,R)$ of the index coding problem, we define by $\lambda(n,q)= \ell (n,q)/n$ the transmission rate of
the optimal solution over an alphabet of size $q$. We also denote by $\lambda^*(n,q)$ the minimum rate achieved by a vector linear solution of block length $n$
over the finite field $\mathbb{F}_q$ of $q$ elements. We are interested in the behavior of $\lambda$ and $\lambda^*$ as  functions of $n$ and $q$.

Let $\mu(\mathcal{I})$ be the maximum  of the total number of messages requested by a set of clients with identical side information, i.e.,
\mbox{$\mu(\mathcal{I})=\max_{Y\subseteq X}|\{x_i; (x_i,Y)\in R\}|$}. Then, it is easy to verify that the optimal rate $\lambda(n,q)$ is lower bounded by
$\mu(\mathcal{I})$, independently of the values of $n$ and $q$. To see this, let \mbox{$Y^*=\arg \max_{Y\subseteq X}|\{x_i;\ (x_i,Y)\in R\}|$} and \mbox{$W=\{x_i; (x_i,Y^*)\in R\}$} and remove all clients that do not have the set $Y^*$ as side information. We note that, since $Y^*\cap W=\emptyset$, the minimum transmission rate of the resulting instance is equal to $|W|=\mu(\mathcal{I})$. Since the rate of the resulting instance is lower or equal to $\lambda(n,q)$ it holds that $\lambda(n,q)\geq \mu(\mathcal{I})$.


%
%
%
%

\begin{definition}
Let $\mathcal{I}(X,R)$ be an instance of the index coding problem. Then, an index code for $\mathcal{I}(X,R)$ that achieves  $\lambda(n,q)=\mu(\mathcal{I})$ is referred to as a \emph{perfect index code}.
\end{definition}

Note that the index code for the example in Figure \ref{fig:butterfly} is not perfect, since, in that case, $\lambda=2$ and $\mu=1$.

\subsection{Network Coding}
Let $G(V,E)$ be a directed acyclic graph with vertex set  $V$  and edge set $E\subset V\times V$. For each edge $e(u,v)\in E$, we define the in-degree of $e$ to be the in-degree of its tail
node $u$, and its  out-degree to be the out-degree of its head node $v$. Furthermore,  we
define $\mathcal{P}(e)$ to be the set of the parent edges of $e$, i.e.,
\mbox{$\mathcal{P}(e(u,v))=\{(w,u);\ (w,u)\in E)\}$}.
%
Let $S\subset E$  be the subset of edges in $E$ of zero in-degree and let $D\subset E$ be the subset of edges of zero out-degree. We refer to edges in $S$ as \emph{input} edges, and those in $D$ as  \emph{output} edges. Also, we define  $m=|E|$ to be the total number of edges, $k=|S|$ be the total number of input edges, and $d=|D|$ be the total number of output edges. Moreover,  we assume that the edges in $E$ are indexed from 1 to $m$ such that  $S =\{e_1,\dots,e_k\}$ and $D=\{e_{m-d+1},\dots,e_m\}$.

We model  a coding network  by a pair $\mathcal{N}(G(V,E),\delta)$ formed by a graph $G(V,E)$ and an onto function \mbox{$\delta: D\longrightarrow S$} from the set of output  edges to the set of input edges. We assume that the tail node of each input edge $e_i$, $i=1,\dots, k$ holds message $x_i$, also denoted as $x(e_i)$.  Each message $x_i$ belongs to a certain alphabet $\Sigma^n$, for a positive integer  $n$.
The edges of the graph represent communication links of unit capacity, i.e., each link can transmit one message per channel use.
 The function $\delta$ specifies for each output edge $e_i$, $i=m-d+1,\dots,m$, the source message $x(\delta(e_i))$ required by its head node. We refer to $\delta$ as the \emph{demand function}.  In a vector solution, each message $x_i$ be divided into $n$ packets $(x_{i1},\dots,x_{in})\in \Sigma^n$. We also denote by $\xi=(x_{11},\dots,x_{1n},\dots,x_{k1},\dots,x_{kn})\in \Sigma^{nk}$ the concatenation of all packets at the input edges.

\begin{definition}[Network Code]
 A  $q$-ary \emph{network code} of block length $n$, or an $(n,q)$ network code, for the network $\mathcal{N}(G(V,E),\delta)$ is a collection
 $$\mathcal{C}=\{f_e=(f_e^1,\dots,f_e^n); e\in E,f_e^i:\Sigma^{nk}\longrightarrow \Sigma,  i=1,\dots,n\},$$
of  functions, called \emph{global encoding} functions, indexed by the edges of $G$, that satisfy, for all $\xi\in \Sigma^{nk}$, the following conditions:
 \begin{enumerate}
   \item [(N1)] $f_{e_i}(\xi)=x_i$, for $i=1,\dots,k$;
   \item [(N2)]$f_{e_i}(\xi)=x(\delta(e_i))$, for $i=m-d+1,\dots,m$;
   \item [(N3)] For each $e=(u,v)\in E\setminus S$ with $\mathcal{P}(e)=\{e_1,\dots,e_
   {p_e}\}$, there exists a function \mbox{$\phi_e:\Sigma^{np_e}\longrightarrow
       \Sigma^n$},  referred to as the \emph{local encoding function} of $e$, such that \mbox{$f_e(\xi)=\phi_{e}(f_{e_1}(\xi),\dots,f_{e_{p_e}}(\xi))$},
       where $p_e$ is the in-degree of $e$, and $\mathcal{P}(e)$ is the set
       of parent edges of $e$.
 \end{enumerate}
\end{definition}


When $n=1$, the network code is referred to as a \emph{scalar} network code. Otherwise, when $n>1$, it is called a \emph{vector} or a \emph{block} network code.
We are  mostly interested here in linear network codes where $\Sigma$ is a finite field $\mathbb{F}$, and all the global and local encoding functions are linear functions of the packets $x_{ij}$. Note that, a scalar linear network code over $GF(p^n)$ will naturally induce a vector linear network code of block length $n$ over $GF(p)$; however, the converse is not necessarily  true.

\section{Connection to Network Coding}\label{sec:connec}

\begin{figure}[t]%
\begin{center}%
\includegraphics[width=3in]{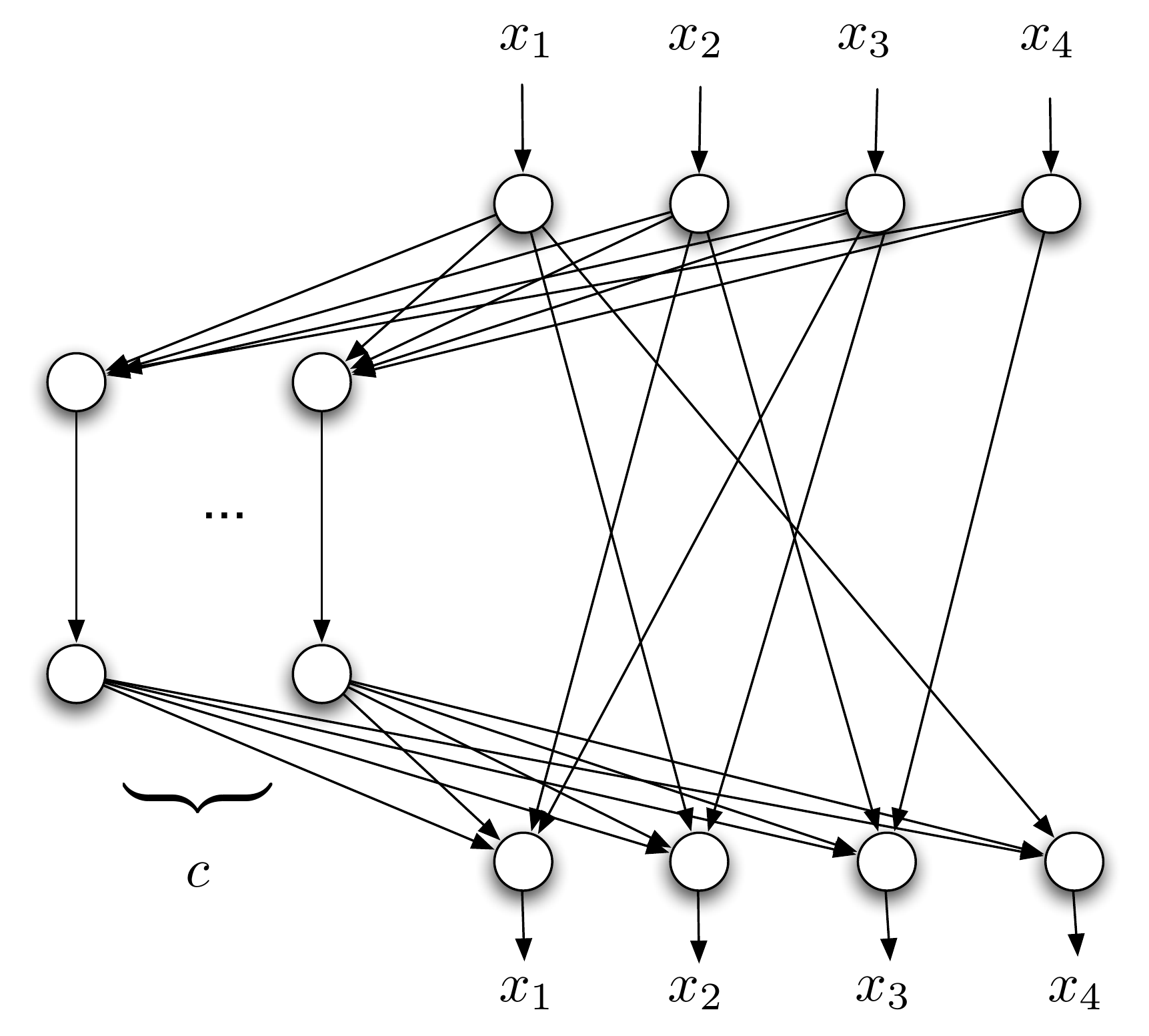}
\end{center}
 \caption{An instance to the network coding problem equivalent to the instance of the index coding problem depicted in Figure \ref{fig:butterfly}.}
\label{fig:indexnet}
\end{figure}

We first note that  network coding is a more general problem than  index coding. Indeed, for every instance of the Index Coding problem and a given integer $c$,  there exists a corresponding instance of the network coding problem that has an  $(n,q)$ network code solution if and only if there exists an $(n,q)$ index code of length $c.n$. For example, Figure \ref{fig:indexnet} depicts the instance of the network coding problem  that corresponds to the instance of the index coding problem presented in Figure~\ref{fig:butterfly}, where the broadcast channel is represented by $c$ ``bottleneck'' edges.

In this section we present a reduction from the network coding problem to the index coding problem showing that two problems  are equivalent. Specifically, for each instance
$\mathcal{N}(G(V,E), \delta)$ of the network coding problem, we construct a corresponding instance $\mathcal{I}_{\mathcal{N}}(Y,R)$ of the index coding problem,
such that $\mathcal{I}_{\mathcal{N}}$ has an $(n,q)$  perfect linear index code  if and only if there exists an $(n,q)$ linear network for $\mathcal{N}$.

\begin{definition}\label{def:1}
Let $\mathcal{N}(G(V,E),\delta)$  be an instance of the Network Coding
problem. We form an instance $\mathcal{I}_{\mathcal{N}}(Y,R)$ of the Index
Coding problem as follows:
\begin{enumerate}

  \item The set of messages $Y$ includes a message $y_i$ for each edge $e_i\in E$
      and all the  messages $x_i\in X$, i.e., \mbox{$Y=\{y_1,\dots,y_m\}\cup
      X$};
  \item The set of clients $R= R_1 \cup\dots\cup R_5$ defined as
      follows:
  \begin{enumerate}
  \item  $R_1=\{ (x_i,\{ y_i\}); e_i\in S\}$
  \item  $R_2=\{ (y_i,\{ x_i\});e_i\in S\}$
  \item $R_3=\{(y_i, \{ y_j; e_j\in\mathcal{P}(e_i)\}); e_i\in
      E\setminus S\}$
  \item $R_4=\{(x(\delta(e_i)),\{y_i\}); e_i\in D\}$
  \item $R_5=\{ (y_i, X); i=1,\dots,m \}$
  \end{enumerate}

\end{enumerate}
\end{definition}

It is easy to verify that  instance $\mathcal{I}_{\mathcal{N}}(Y,R)$ satisfies $\mu(\mathcal{I}_{\mathcal{N}})=m$.

\begin{theorem}\label{th:equivalence}
Let $\mathcal{N}(G(V,E),X,\delta)$ be an instance of the network coding problem,
and let $\mathcal{I}_{\mathcal{N}}(Y,R)$ be the corresponding instance of the index coding problem, as defined above. Then, there exists an $(n,q)$ perfect linear  index code  for
$\mathcal{I}_{\mathcal{N}}$,   if and only if, there exists a linear $(n,q)$ network code for $\mathcal{N}$.
\end{theorem}
\begin{proof}
Suppose there is a linear $(n,q)$ network code $ C=\{f_e(X);
f_e:(\mathbb{F}_q^n)^k\rightarrow \mathbb{F}_q^n, e\in E\}$ for $\mathcal{N}$ over the finite field $\mathbb{F}_q$ of size $q$ for some integer $n$.

Define $g:(\mathbb{F}_q^n)^{m+k}\rightarrow (\mathbb{F}_q^n)^m$ such that $\forall
Z=(x_1,\dots,x_k,y_1,\dots,y_m)\in (\mathbb{F}_q^n)^{m+k}, g(Z)=(g_1(Z),\dots,g_m(Z))$
where $g_i(Z)=y_i+f_{e_i}(X)$, $i=1,\dots, m$. More specifically, we have

\begin{align*}
 g_i(Z)&=y_i+x_i\quad  &i=1,\dots,k,\\
g_i(Z)&=y_i+f_{e_i}(X) \quad &i=k+1,\dots,m-d,\\
g_i(Z)&=y_i+x(\delta(e_i)) \quad &i=m-d+1,\dots,m.
\end{align*}

Next, we show that $g(Z)$ is in fact an index code for
$\mathcal{I}_{\mathcal{N}}$ by proving the existence of the decoding functions. We consider the following five cases:
\begin{enumerate}
  \item $\forall \rho=(x_i,\{y_i\})\in R_1, \psi_{\rho}=g_i(Z)-y_i$,
  \item $\forall \rho=(y_i,\{x_i\})\in R_2, \psi_{\rho}=g_i(Z)-x_i$,
  \item $\forall \rho=(y_i, \{y_{i_1},\dots, y_{i_p}\})\in R_3$, since $C$ is
      a linear network code for $\mathcal{N}$, there exists a linear function $\phi_{e_i}$
      such that $f_{e_i}(X)=\phi_{e_i}(f_{e_{i_1}}(X),\dots,
      f_{e_{i_p}}(X))$. Thus,
      \mbox{$\psi_{\rho}=g_i(Z)-\phi_{e_i}(g_{{i_1}}(Z)-y_{i_1},\dots,g_{{i_p}}(Z)-y_{i_p})$},
   \item $\forall \rho=(x(\delta(e_i)),\{y_i\})\in R_4, e_i\in D,
       \psi_{\rho}=g_i(Z)-y_i$,
   \item $\forall \rho=(y_i,X)\in R_5,\psi_{\rho}=g_i(Z)-f_{e_{i}}(X)$.
\end{enumerate}

To prove the converse, we  assume that $g: (\mathbb{F}_q^n)^{m+k} \longrightarrow (\mathbb{F}_q^n)^m$ is a perfect linear
$(n,q)$ index code for $I_{\mathcal{N}}$ over the field $\mathbb{F}_q$. Again, we denote
$Z=(x_1,\dots,x_k,y_1,\dots,y_m)\in (\mathbb{F}_q^n)^{m+k},$ and
 $g(Z)=(g_1(Z),\dots,g_m(Z)), $ $x_i,y_i$ and  $g_i(Z)\in \mathbb{F}_q^n.$ We also
 write
 $$g_i(Z)=\sum_{j=1}^{k}x_jA_{ij}+\sum_{j=1}^m y_jB_{ij},$$

for  $i=1,\dots, m$, and $A_{ij}, B_{ij}\in M_{\mathbb{F}_q}(n,n)$, where $M_{\mathbb{F}_q}(n,n)$ is  the  set of $n\times n$ matrices with elements in $\mathbb{F}_q$.

The functions $\psi_\rho$  exist for all $\rho\in R_5$ if and only if the matrix
$M=[B_{ij}]\in  M_{\mathbb{F}_q}(nm,nm)$, which has  the matrix $B_{ij}$ as a
block submatrix in the $(i,j)$th position,  is invertible. Define $h:
(\mathbb{F}_q^n)^{m+k} \longrightarrow (\mathbb{F}_q^n)^m$, such that $h(Z)=g(Z)M^{-1},
\forall Z\in (\mathbb{F}_q^n)^{m+k} $. So, we obtain
$$h_i(Z)=y_i+\sum_{j=1}^kx_j C_{ij}, i=1,\dots, m,$$
where $C_{ij}\in  M_{\mathbb{F}_q}(n,n)$.
We note that  $h(Z)$ is a valid index code for $\mathcal{I}_{\mathcal{N}}$.  In
fact, $\forall \rho=(x,H)\in R$ with $\psi_\rho(g,(z)_{z\in H})=x$,
$\psi_\rho^\prime(h,(z)_{z\in H})=\psi_\rho(hM,(z)_{z\in H}))$ is a valid decoding function corresponding to the client $\rho$ and the index code $h(Z)$.

For all $\rho\in\ R_1\cup R_4$, $ \psi_\rho^\prime$ exists  iff for
$i=1,\dots,k,m-d+1,\dots, m, j=1\dots k$ and $j\neq i$ it holds that  $C_{ij}=[0]\in
M_{\mathbb{F}_q}(n,n)$ and $C_{ii}$ is invertible, where $[0]$ denotes the all
zeros matrix. This implies that
\begin{equation}\label{eq:indexcod}
\begin{split}
 h_i(Z)&=y_i+ x_iC_{ii}, i=1,\dots, k\\
    h_i(Z)&=y_i+\sum_{j=1}^k x_j C_{ij}, i=k+1,\dots, m-d\\
    h_i(Z)&=y_i+ x(\delta(e_i))C_{ii}, i=m-d+1,\dots,m
\end{split}
\end{equation}
Next, we define the functions $f_{e_i}:(\mathbb{F}_q^n)^k \longrightarrow \mathbb{F}_q^n,
e_i\in E$ as follows:
\begin{enumerate}
\item $f_{e_i}(X)=x_i$, for $ i=1,\dots,k$
\item $f_{e_i}(X)=\sum_{j=1}^k x_jC_{ij}$, for $ i=k+1,\dots,m-d$
\item $f_{e_i}(X)=x(\delta(e_i))$, for $ i=m-d+1,\dots, m$.
\end{enumerate}

We will show that  $C=\{f_{e_i}; e_i\in E\}$ is a linear $(n,q)$ network code for
$\mathcal{N}$ by proving that condition N3 holds.

Let $e_i$ be an edge in $E\setminus S$ with  the set of parent edges
$\mathcal{P}(e_i)=\{e_{i_1},\dots,e_{i_p}\}$ . We denote by
$I_i=\{i_1,\dots,i_p\}$
and $\rho_i=(y_i,\{y_{i_1},\dots,y_{i_p}\})\in R_3$. Then, there is a linear function
$\psi_{\rho_i}^\prime$ such that $y_i= \psi_{\rho_i}^\prime
(h_1,\dots,h_m,y_{i_1},\dots, y_{i_p})$. Hence, there exist matrices $T_{ij},
T^\prime_{i\alpha}\in M_{\mathbb{F}_q}(n,n)$ such that
\begin{equation}\label{eq:int}
y_i=\sum_{j=1}^mh_j T_{ij}+\sum_{\alpha\in I_i}y_\alpha T^\prime_{i\alpha}
\end{equation}
 Substituting the expressions of the $h_j$'s given by Eq.~\eqref{eq:indexcod} in Eq.~\eqref{eq:int}, we get that the following:

\begin{itemize}
  \item $T_{ii}$ is the  identity matrix,
  \item $T^\prime_{i\alpha}=-T_{i\alpha}\forall \alpha\in I_i$,
  \item $T_{ij}=[0]\ \forall j\notin I_i\cup\{i\}$.
\end{itemize}

%
  Therefore, we obtain $$f_{e_i}=-\sum_{\alpha\in I_i} f_{e_\alpha  }T_{i\alpha}, \forall e_i\in E\setminus S,$$ and $C$ is a feasible network code for $\mathcal{N}$.
\end{proof}

\begin{lemma}\label{ref:nonlinear}
Let $\mathcal{N}(G(V,E),\delta)$ be an instance of the Network Coding problem,
and let $\mathcal{I}_{\mathcal{N}}(Y,R)$ be the corresponding index problem. If there is an
$(n,q)$ network code (not necessarily linear)
 for $\mathcal{N}$, then there is a perfect $(n,q)$ index code
for $\mathcal{I}_{\mathcal{N}}$.
\end{lemma}
\begin{proof}
Suppose there is an $(n,q)$ network code $ C=\{f_e(X); f_e:(\Sigma^n)^k\rightarrow \Sigma^n, e\in E\}$ for $\mathcal{N}$ over the q-ary alphabet $\Sigma$. Without loss of generality, we assume that $\Sigma=\{0,1,\dots,q-1\}$.

Define $g:(\Sigma^n)^{m+k}\rightarrow (\Sigma^n)^m$ such that $\forall
Z=(x_1,\dots,x_k,y_1,\dots,y_m)\in (\Sigma^n)^{m+k}, g(Z)=(g_1(Z),\dots,g_m(Z))$
with

$$g_i(Z)=y_i+ f_{e_i}(X), \quad i=1,\dots,m$$

where ``+'' designates addition in $GF(q)^n$. Then, the same argument of the previous proof holds similarly here, and $g$ is an index code for $\mathcal{I}_{\mathcal{N}}$.
\end{proof}

\section{Connection To Matroid Theory}\label{sec:Indexoid}

A matroid $\mathcal{M}(Y,r)$ is a couple formed by a  set  $Y$ and a
function \mbox{$r:2^Y\longrightarrow \mathbb{N}_0$}, where $2^Y$ is the power
set of $Y$ and $\mathbb{N}_0$ is the set of non-negative integer numbers $\{0,1,2,\dots\}$, satisfying the following three conditions:

\begin{enumerate}
  \item[(M1)] $r(A)\leq |A|$  for $\forall A\subseteq Y$;
  \item [(M2)]$r(A)\leq r(B)$ for $\forall A\subseteq B\subseteq Y$;
  \item [(M3)]$r(A\cup B)+r(A\cap B)\leq r(A)+r(B)$ for $\forall A,B\subseteq Y.$
\end{enumerate}

The set $Y$ is called the \emph{ground set} of the matroid $\mathcal{M}$.
The function $r$ is called the \emph{rank function} of the matroid. The rank $r_\mathcal{M}$ of  the matroid $\mathcal{M}$ is defined as $r_\mathcal{M}=r(Y)$.

We refer to $B\subseteq Y$ as an \emph{independent} set if $r(B)=|B|$,  otherwise, it is referred to as a  \emph{dependent} set. A maximal independent set is referred to  as a \emph{basis}. It can be shown that all bases in a matroid have the same number of elements. In fact, for any basis $B$, it holds that $r(B)=|B|=r_\mathcal{M}$. A minimal dependent subset $C\subseteq Y$ is referred to as a \emph{circuit}. For each element  $c$ of $C$ it holds that $r(C\setminus\{c\})=|C|-1=r(C)$.
We define  $\mathfrak{B}(\mathcal{M})$ to be the set of all the bases  of the matroid $\mathcal{M}$, and $\mathfrak{C}(\mathcal{M})$ be the set of all circuits of $\mathcal{M}$.

Matroid theory is a well studied topic in discrete mathematics. References \cite{OX93} and \cite{Welsh76} provide a comprehensive discussion
on this subject. Linear and multilinear representations of matroids over finite fields are major topics in  matroid theory (see \cite[Chapter 6]{OX93},  \cite{SA98}, and \cite{M99}).

\begin{definition}

Let $Y=\{y_1,\dots,y_m\}$ be a set  whose elements are indexed by the integers from 1 to $m$. For any collection of $m$ matrices $M_1,\dots,M_m \in\mathbb{M}_{\mathbb{F}}(n,k)$, and any subset  $I=\{y_{i_1},\dots,y_{i_\delta}\}\subseteq Y$, with $i_1<\dots<i_\delta$, define $$M_I=[M_{i_1}|\dots|M_{i_\delta}]\in \mathbb{M}_{\mathbb{F}}(n,\delta k).$$
 That is the matrix $M_I$  obtained by concatenating matrices
$M_{i_1},\dots,M_{i_\delta}$ from left to right in the increasing order of the indices $i_1,\dots,i_\delta$.
\end{definition}


\begin{definition}
Let $\mathcal{M}(Y,r)$ be a matroid of rank \mbox{$r_\mathcal{M}=k$} on the ground set $Y=\{y_1,\dots,y_m\}$. The matroid $\mathcal{M}$ is said to have a multilinear representation of dimension $n$, or an $n$-linear representation, over a field $\mathbb{F}$, if
there exist matrices $M_1,\dots, M_m\in \mathbb{M}_{\mathbb{F}}(kn,n)$ such that, $\forall I\subseteq Y,$
\begin{equation}\label{eq:rep}
 \rank (M_I)=n\cdot r(I).
\end{equation}
\end{definition}

Linear representation corresponding to the  case of $n=1$ is the most studied case in matroid theory, see for example \cite[Chapter 6]{OX93}. Multilinear representation is a generalization of this concept from vectors  to vector spaces, and was discussed in \cite{SA98,M99}.

\begin{figure}[t]
\begin{center}
	\includegraphics[width=2.5in]{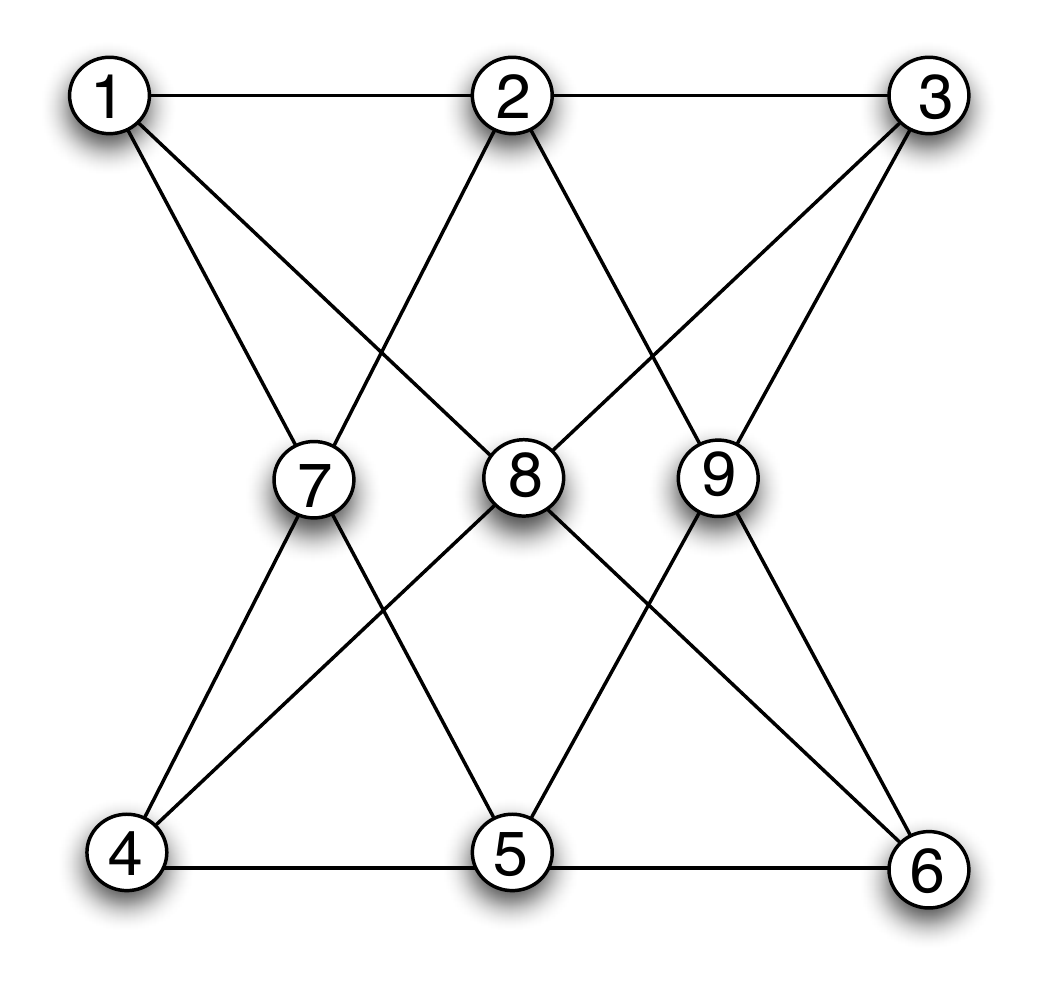}
 \caption{ A graphical representation of the non-Pappus matroid  of rank 3 \cite[p.43]{OX93}. Cycles are
 represented by straight lines.} \label{fig:nonPappusMat}
\end{center}
\end{figure}

\begin{example}\label{ex:uniform}
 The uniform matroid $U_{2,3}$ is defined on a ground set $Y=\{y_1,y_2,y_3\}$ of three elements, such that $\forall I\subseteq Y$ and $|I|\leq 2,r(I)=|I|$, and $r(Y)=2$. It is easy to verify that matrices $M_1=[0\quad 1]^T, M_2=[0\quad 1]^T, M_3=[1\quad 1]^T$ form a linear representation of $U_{2,3}$ of dimension 1 over any field. This will automatically induce a multi-linear representation  of dimension 2, for instance, of  $U_{2,3}$  over any field:
 \begin{equation*}
   M_1=\begin{bmatrix}
          1 & 0 \\
          0 & 0 \\
          0 & 1 \\
          0 & 0 \\
        \end{bmatrix},
    M_2=\begin{bmatrix}
          0 & 0 \\
          1 & 0 \\
          0 & 0 \\
          0 & 1 \\
        \end{bmatrix},
        M_3=\begin{bmatrix}
          1 & 0 \\
          1 & 0 \\
          0 & 1 \\
          0 & 1 \\
        \end{bmatrix}
 \end{equation*}
\end{example}

\begin{example}\label{ex:non-Pappus} The non-Pappus matroid (see e.g., \cite[\S 1.5]{OX93}) $\mathcal{M}_{np}(Y,r)$ is defined over a ground set $Y=\{y_1,\dots,y_9\}$ and can be represented geometrically as shown in Figure \ref{fig:nonPappusMat}. Let $Y_0=\{\{1,2,3\},\{1,5,7\},\{3,5,9\},\{2,4,7\} ,\{4,5,6\},$ $\{2,6,9\},\{1,6,8\},\{3,4,8\} \}$. The rank function of the non-Pappus matroid is  $r(I)=\min (|I|,3)$  $\forall I\in 2^Y\setminus Y_0$, and $r(I)=2$ for $\forall I\in Y_0$. Note that $Y_0$ is the set of circuits of the non-Pappus matroid.

It is known from Pappus theorem \cite[p.173]{OX93}, that the  non-Pappus matroid is not linearly representable over any field. However, it was shown in \cite{SA98} and \cite{M99}, that  it has a 2-linear representation over $GF(3)$, given below by the following $6\times 2$ matrices $M_1,\dots,M_9$:
\small
\begin{equation}\label{eq:NPrep}
   [M_1|\dots|M_9]=\begin{pmatrix}
                     10 & 10 & 00 &  10& 00 & 10 & 10 & 10&  00\\
                     01 & 01 & 00 & 01 & 00 & 01 & 01 & 01 & 00 \\
                     00 & 00 & 00 & 10 & 10 & 21 & 01 & 10 & 10 \\
                     00 & 00 & 00 & 02 & 01 & 20 & 12 & 02 & 01 \\
                     00 & 10 & 10 & 01 & 00 & 01 & 00 & 11 & 10 \\
                     00 & 01 & 01 & 21 & 00 & 21& 00 & 10 & 01 \\
                   \end{pmatrix}.
\end{equation}
\normalsize
\end{example}

In the rest of this section, we present a reduction from the matroid representation problem to the index coding problem.

\begin{definition} \label{def:matindex}
Given a matroid $\mathcal{\mathcal{M}}(Y,r)$ of rank $k$ over ground set
$Y=\{y_1,\dots, y_m\}$, we define the corresponding Index Coding
problem $\mathcal{I}_\mathcal{\mathcal{M}}(Z,R)$ as follows:
\begin{enumerate}
  \item $Z=Y\cup X$, where $X= \{x_1,\dots,x_k\}$,
   \item $R=R_1\cup R_2\cup R_3$ where
  \begin{enumerate}
  \item  $R_1=\{ (x_i,B);  B\in \mathfrak{B}(\mathcal{M}), i=1,\dots,k\}$
  \item  $R_2=\{ (y,C\setminus\{y\});C\in \mathfrak{C}(\mathcal{M}), y\in C\}$
  \item $R_3=\{(y_i,X); i=1,\dots,m\}$
  \end{enumerate}
\end{enumerate}
\end{definition}
Note that $\mu(\mathcal{I}_\mathcal{M})=m.$

\begin{theorem}\label{th:equivalence2}
Let $\mathcal{M}(Y,r)$ be a matroid on the set \mbox{$Y=\{y_1,\dots,y_m\}$}, and $\mathcal{I}_{\mathcal{M}}(Z,R)$ be its corresponding Index Coding problem.  Then, the matroid $\mathcal{M}$ has an $n$-linear representation over $\mathbb{F}_q$ if and only if there exists a perfect linear $(n,q)$ index code for $\mathcal{I}_{\mathcal{M}}$.
\end{theorem}
\begin{proof}
First, we assume that in $\mathcal{I}_{\mathcal{M}}(Z,R)$ all messages are split into $n$ packets, and we write $y_i=(y_{i1},\dots,y_{in})$, $x_i=(x_{i1},\dots,x_{in})\in \mathbb{F}_q^n$, $\xi=(x_{11},\dots,x_{1n},\dots,x_{k1},\dots,x_{kn})\in \mathbb{F}_q^{kn}$, and $\chi=(y_{11},\dots,y_{1n},\dots,y_{m1},\dots,y_{mn},$ $x_{11},\dots,x_{1n},\dots,x_{k1},\dots,x_{kn})\in \mathbb{F}_q^{(m+k)n}.$

Let $M_1,\dots,M_m \in \mathbb{M}_{\mathbb{F}_q}(kn,n)$ be an $n$-linear representation of the matroid $\mathcal{M}$. Consider the following linear map $f(\chi)=(f_1(\chi),\dots, f_m(\chi))$
$$f_i(\chi)=y_i+\xi M_i \in \mathbb{F}_q^n, i=1,\dots,m.$$

We claim that $f$ is a perfect $(n,q)$ linear index code for $\mathcal{I}_{\mathcal{M}}$. To this end, we show the existence of the decoding functions of condition (I1) for all the clients in $R$:
\begin{enumerate}
  \item Fix a  basis $B=\{y_{i_1},\dots,y_{i_k}\}\in\mathfrak{B}(\mathcal{M})$, with $i_1<i_2<\dots<i_k$, and let $\rho_i=(x_i,B)\in R_1$,  $i=1,\dots,k$. By Eq.~\eqref{eq:rep} $\rank (M_B)=kn$, hence the  $kn\times kn$ matrix $M_B$ is invertible. Thus, the corresponding decoding functions can be written as
       $$\psi_{\rho_i}=[f_{i_1}-y_{i_1}|\dots|f_{i_k}-y_{i_k}]U_i,$$
       where the $U_i$'s are the $kn\times n$ the block matrices that form $M_B^{-1}$ in the following way:
       $$[U_i|\dots|U_k]=M_B^{-1}.$$

  \item Let $C=\{y_{i_1},\dots,y_{i_c}\}\in \mathfrak{C}(\mathcal{M})$, with $i_1<i_2<\dots<i_c$, and $\rho=(y_{i_1},C')\in R_2$, with $C'=C-y_{i_1}$. We have $\rank(M_{C'})=\rank(M_C)$ by the definition of matroid cycles. Therefore, there is a matrix $T\in\mathbb{M}_{\mathbb{F}_q}(cn-n,n)$, such that, $M_{i_1}=M_{C'}T$. Now, note that       $$[f_{i_2}-y_{i_2}|\dots|f_{i_c}-y_{i_c}]=\xi M_{C'}.$$
      Therefore, the corresponding decoding function is
      $$\psi_\rho=f_{i_1}-[f_{i_2}-y_{i_2}|\dots|f_{i_c}-y_{i_c}]T.$$
  \item For all $\rho=(y_i,X)\in R_3, \psi_\rho(f,\xi)=f_i-\xi M_i.$
\end{enumerate}
Since this index code  satisfies the lower bound $\mu(\mathcal{I}_{\mathcal{M}})=m$, it is a perfect index code.

Now, suppose that $f(\chi)=(f_1(\chi),\dots,f_m(\chi))$, $f_i(\chi)\in \mathbb{F}_q^n$, is a perfect $(n,q)$ linear index code for $\mathcal{I}_{\mathcal{M}}$. We will show that this will induce an  $n$-linear representation of the matroid $\mathcal{M}$  over $\mathbb{F}_q$.

Due to the clients in $R_3$, we can use the same reasoning used in the proof of the converse of Theorem \ref{th:equivalence} and assume that the functions $f_i(\chi), i=1,\dots,m$, have the following form
\begin{equation}\label{eq:IC}
f_i(\chi)= y_i+ \xi A_i,
\end{equation}
where the $A_i$'s are $kn\times n$ matrices over $\mathbb{F}_q$. We claim that these matrices form an $n$-linear representation of $\mathcal{M}$ over $\mathbb{F}_q$. To prove this, it suffices to show that the matrices $A_i$'s satisfy   Eq.~\eqref{eq:rep} for all the bases and cycles of $\mathcal{M}.$

Let $B\in \mathfrak{B}(\mathcal{M})$ a basis. Then, by Eq.~\eqref{eq:IC}, the clients $(x_j,B), j=1,\dots,k$, will be able to decode their required messages iff $A_B$ is invertible. Therefore, $\rank(A_B)=nk=n r(B).$

Let  $C\in \mathfrak{C}(\mathcal{M})$ a circuit. Pick $y_{i_1}\in C$ let $C'=C-y_{i_1}$. We have  $r(C')=|C|-1=|C'|$,i.e., $C'$ is an independent set of the matroid, and  there  is a basis $B$ of $\mathcal{M}$ such that $C'\subseteq B$ (by the independence augmentation axiom  \cite[chap. 1]{OX93}). Thus, from the previous discussion, $A_{C'}$ has full rank, i.e. $\rank(A_{C'})=(|C|-1)n$.
Now consider the client $\rho=(y_{i_1},C')\in R_2$, the existence of the corresponding linear decoding function $\psi_\rho$ implies that there exists a matrix $T\in  \mathbb{M}_{\mathbb{F}}(|C|n-n,n)$ such that $A_{i_1}=A_{C'}T.$ So, $\rank(A_C)=\rank(A_{C'})=n(|C|-1)=nr(C).$

\end{proof}

\section{Properties of Index Codes}\label{sec:applications}
\subsection{Block Encoding}

Index coding, as previously noted, is related to the problem of  zero-error source coding with side information, discussed by Witsenhausen  in
\cite{W76}. Two cases were studied there, depending on whether the transmitter knows the side information available to the receiver or not. It was  shown  that in the former case the repeated scalar encoding is optimal, i.e., block encoding does not have any advantage over the scalar encoding. We will demonstrate in this section that this result does not always hold for the index coding problem, which can be seen as an extension of the point to point problem discussed in \cite{W76}.

\begin{figure}[t]
\begin{center}
	\includegraphics[width=3in]{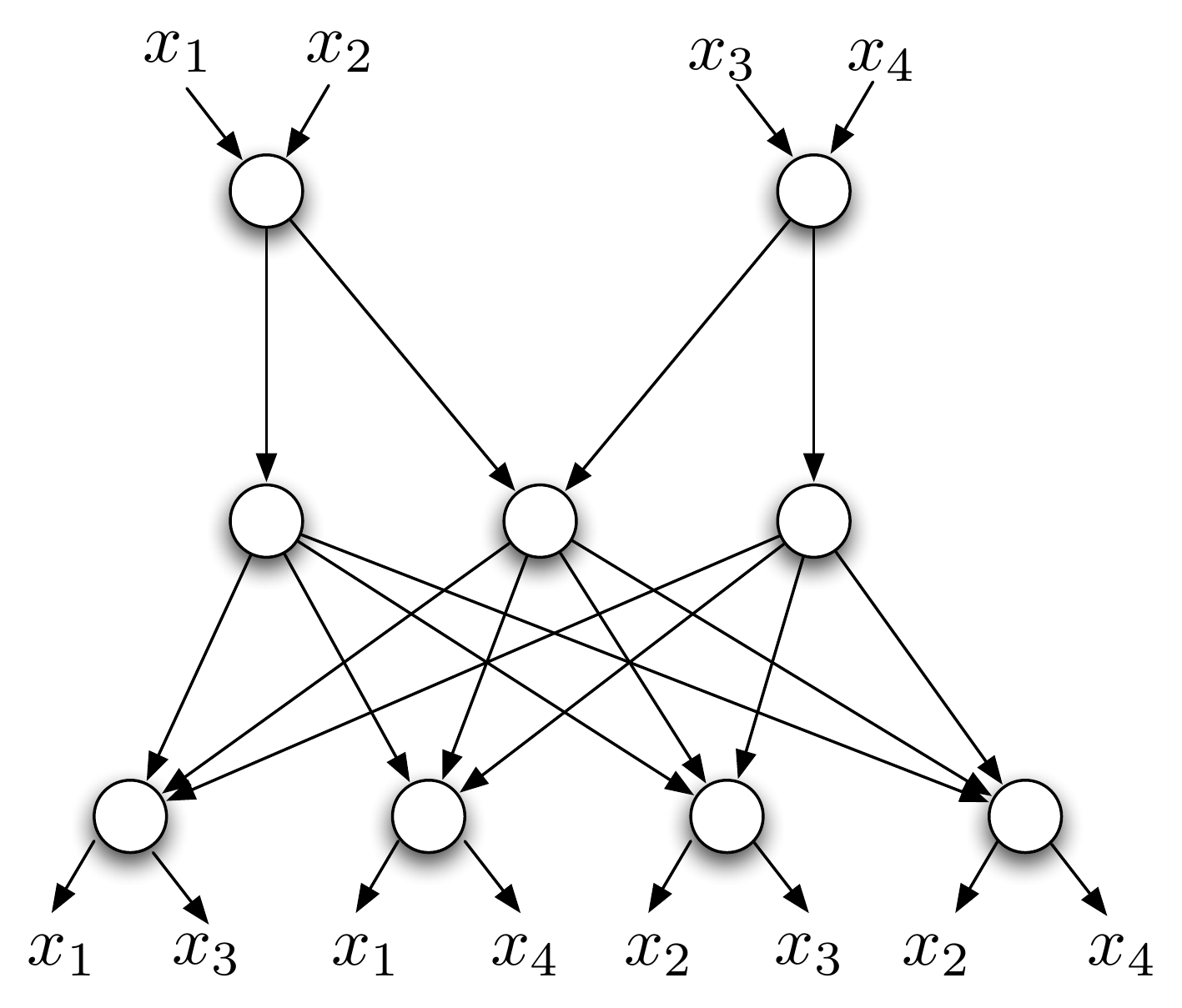}
 \caption{The M-Netwrok $\mathcal{N}_1$ \cite{MEHK03}. } \label{fig:MPappus}
\end{center}
\end{figure}

Let $\mathcal{N}_1$ be the M-network introduced in \cite{MEHK03} and depicted
in Figure \ref{fig:MPappus}. It was shown in \cite{DFZ07} that this network
does not have a scalar linear network code, but has a vector linear
one of block length $2$. Interestingly, such a vector linear solution does
not require encoding and can be solved using a forwarding scheme. A more general result was proven in \cite{DFZ07}:

\begin{theorem}\label{th:Mnet}
The M-network has a linear network code of block length $n$ if and only if $n$ is even.
\end{theorem}

Consider an instance $I_{\mathcal{N}_1}$ of the  index coding problem  corresponding  to the M-network  obtained by the construction of Definition \ref{def:1}. By theorem \ref{th:Mnet}, $I_{\mathcal{N}_1}$  does not admit a perfect scalar linear index code. It  has, however, a perfect linear index code of block length 2,
 over any field. Thus, $\mathcal{I}_{\mathcal{N}_1}$ is an instance of the index coding problem where vector linear coding outperforms scalar linear one. This result can be summarized
by the following corollary:
\begin{corollary}
For $\mathcal{I}_{\mathcal{N}_1}, \lambda^*(2,2)<\lambda^*(1,2).$
\end{corollary}
\begin{proof} Follows directly from theorems \ref{th:equivalence} and \ref{th:Mnet}.
\end{proof}

Another similar instance of the index coding problem is $\mathcal{I}_{\mathcal{M}_{np}}$  obtained by applying the construction of Definition \ref{def:matindex} to  the non-Pappus matroid $\mathcal{M}_{np}.$ Since the non-Pappus matroid $\mathcal{M}_{np}$ does not admit any linear representation, by Theorem \ref{th:equivalence2}, there is, also, no perfect scalar linear index code for $\mathcal{I}_{\mathcal{M}_{np}}$. Nevertheless, the multilinear representation of the non-Pappus matroid over $GF(3)$ described in Example~\ref{ex:non-Pappus} induces a perfect $(3,2)$ vector linear index code for $\mathcal{I}_{\mathcal{M}_{np}}$.

\begin{corollary}
For instance $\mathcal{I}_{\mathcal{M}_{np}}$ of the Index Coding problem it holds that $ \lambda^*(2,2)<\lambda^*(1,2).$
\end{corollary}
\begin{proof} Follows directly from Theorem \ref{th:equivalence2}.
\end{proof}

\subsection{Linearity vs. Non-Linearity}
Linearity is a desired property for any code, including index codes.  It was
conjectured in \cite{BBJK06} that binary scalar linear index codes are optimal,
meaning that $\lambda^*(1,2)=\lambda(1,2)$ for all index coding instances. Lubetzky and Stav  disproved this conjecture in \cite{LS07} for the \emph{scalar
linear} case by providing, for any given number of messages  $k$ and field
$\mathbb{F}_q$, a family of  instances of the index coding problem with an increasing gap between $\lambda^*(1,q)$ and $\lambda(1,q)$.

In this section, we show that \emph{vector linear} codes are still suboptimal. In particular, we give an instance where non-linear index codes outperform vector linear codes for any choice of field and block length $n$. Our proof is based on the insufficiency of linear network codes result proved by Dougherty et al. \cite{DFZ05}. Specifically,  \cite{DFZ05} showed that the network $\mathcal{N}_3$, depicted in Figure~\ref{fig:zegernet} has the following property:

\begin{theorem}\label{theor:aux}\cite{DFZ05}
The network $\mathcal{N}_3$ does not have a linear network code, but has a $(2,4)$
non-linear one.
\end{theorem}

Let $\mathcal{I}_{\mathcal{N}_3}$  be the instance of the Index Coding problem that corresponds to $\mathcal{N}_3$, constructed according to
Definition~\ref{def:1}. Theorem~\ref{theor:aux} implies that $\mathcal{I}_{\mathcal{N}_3}$  does not have a perfect linear index code.
However, by Lemma~\ref{ref:nonlinear}, a $(2,4)$ non-linear code of $\mathcal{N}_3$ can be used to construct a $(2,4)$ perfect non-linear index code for
$\mathcal{I}_{\mathcal{N}_3}$ that satisfies $\lambda(2,4)=\mu(\mathcal{I})$. We summarize this result by the following corollary.

\begin{corollary}\label{col:16}
For the instance $\mathcal{I}_{\mathcal{N}_3}$ of the  Index Coding problem, it
holds that $\lambda(2,4)=\mu(\mathcal{I}_{\mathcal{N}_3})<\lambda^*(2,4)$.
\end{corollary}

\begin{figure}[t]
\begin{center}
	\includegraphics[width=3in]{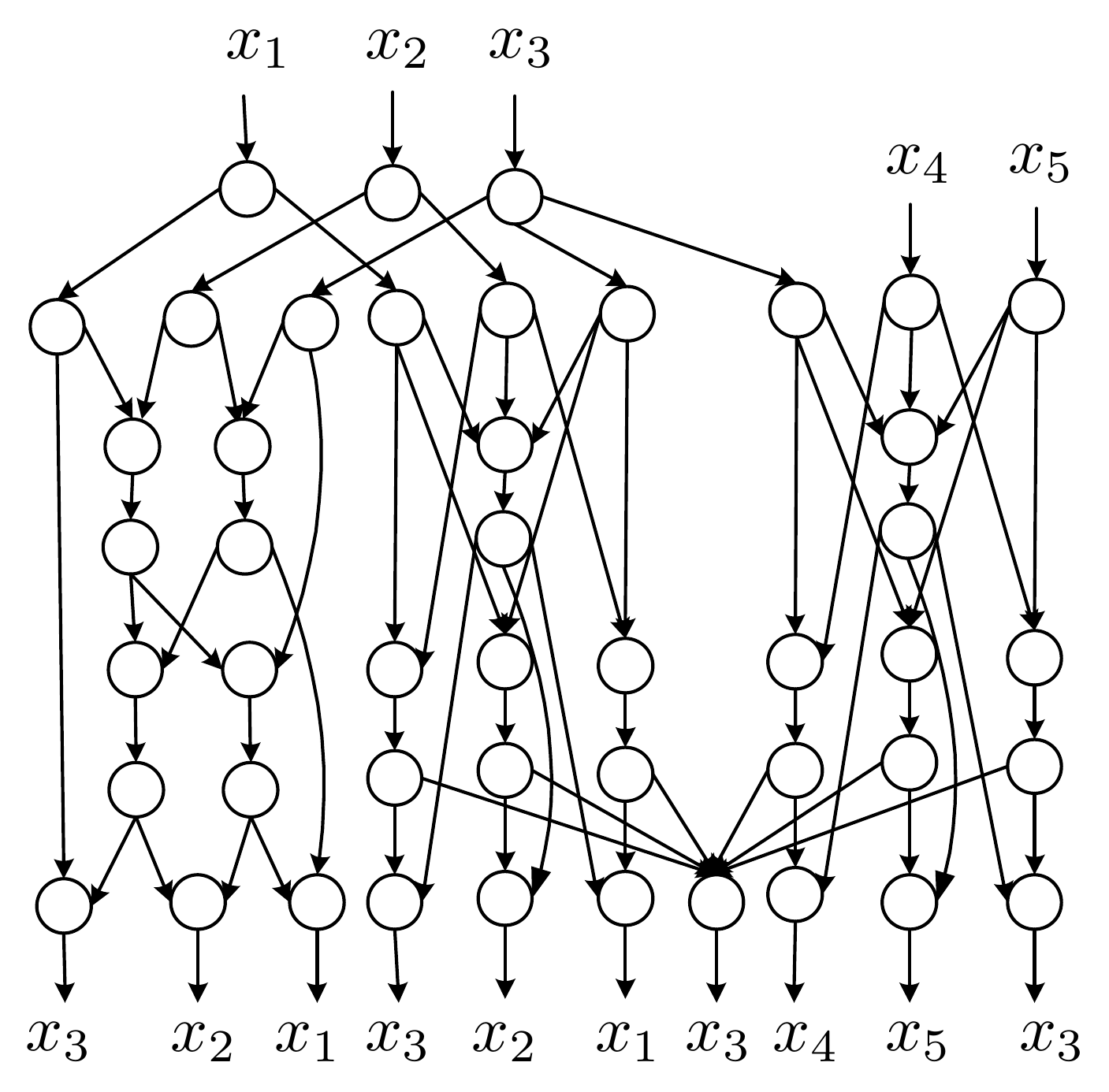}
 \caption{The network $\mathcal{N}_3$ of \cite{DFZ05}. $\mathcal{N}_3$ does not have
 a linear network code over any field, but has a non-linear one over a quaternary alphabet. } \label{fig:zegernet}
\end{center}
\end{figure}

\section{Discussion}\label{sec:discussion}
\subsection{Matroids and Networks}\label{sec:netroid}
\begin{figure*}[t]
\begin{center}
	\includegraphics[width=4in]{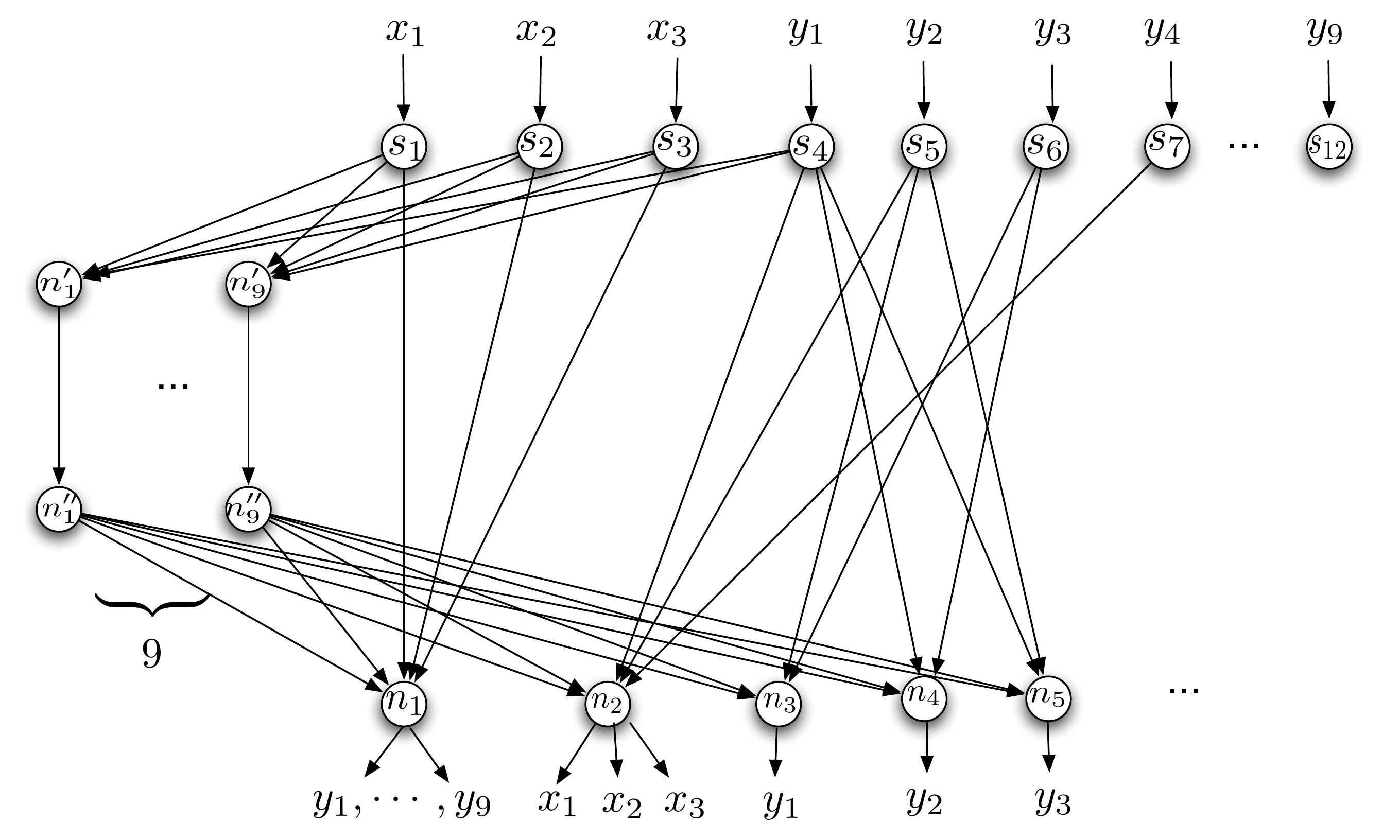}
 \caption{Part of the network resulting equivalent to the non-Pappus matroid resulting from the construction of Definition \ref{def:matnet}. }\label{fig:Pappus2}
\end{center}
\end{figure*}

Dougherty et al. \cite{DFZ05,DFZ07} used results on the representability of matroids for construction of network $\mathcal{N}_3$ depicted in Figure~\ref{fig:zegernet} that  served as a counter example to the conjecture of the sufficiency of linear network codes. They defined also the concept of a matroidal network, and presented  a method for constructing networks from matroids \cite[Section V.B]{DFZ07}. Given a certain matroid, they  design an instance to the network coding problem that forces the same independency relations of the matroid  to exist in  the set of edge messages.  However,  not all of the matroid dependency relations  are reflected in this  network.

In this paper, we present a new construction that remedies this problem. Our construction is based on the reduction presented in Section \ref{sec:Indexoid}. It provides a stronger connection between matroids and network codes. Specifically, for a given matroid, we construct a network such that any multilinear representation of it will induce a linear network code for the obtained network over the same field, and vice versa. This result will permit us to apply many important  results on matroid linear representability to the network coding theory.

\begin{definition}\label{def:matnet}
Let $\mathcal{M}(Y,r)$ be a matroid of rank $k$ defined on the set $Y=\{y_1,\dots,y_m\}$, and $\mathcal{I}_{\mathcal{M}}(Z,R)$ the corresponding index coding problem as described in Definition \ref{def:matindex}.  We associate to it the 6-partite network $\mathcal{N}(\mathcal{I_\mathcal{M}})$ constructed as follows:

\begin{enumerate}
\item $V\supset V_1\cup V_2 \cup V_3$, where $V_1=\{s_1,\dots,s_{m+k}\}$, $V_2=\{n_1',\dots,n_m'\}$, and $V_3=\{n_1'',\dots,n_m''\}$.
    \item Connect each node $s_i, i=1,\dots,k,$  to an input edge carrying an information source $x_i$ at its tail node, and each node $s_i, i=k+1,\dots,m+k,$ to an input edge carrying an information source $y_i$.
  \item Add edges $(s_i,n_j')$ and $(s_i,n_j'')$, for $i=1,\dots,m+k$ and $j=1,\dots,m$.
  \item Add edges $(n_j',n_j'')$ for $j=1,\dots,m$.
  \item  For each client $\rho=(z,H)\in R$, add a vertex $n_\rho$ to the network, and connect it to an output edge that demands source $z$. And, for each $z'\in H$, add edge $(s',n_\rho)$, where $s'\in V_1$ is connected to an input edge carrying source $z'$.
  \item  For each $\rho\in R$, add edge $(n_j'',n_\rho)$, for $j=1,\dots,m$.
\end{enumerate}
\end{definition}


\begin{proposition}\label{prop:matnet}
The matroid $\mathcal{M}$ has an $n$-linear representation over $\mathbb{F}_q$ iff the network $\mathcal{N}(\mathcal{I_\mathcal{M}})$ has an $(n,q)$ linear network code.
\end{proposition}
\begin{proof}
It can be easily seen that any $(n,q)$ perfect linear index code for $\mathcal{I_\mathcal{M}}$ will imply an $(n,q)$ linear network code for $\mathcal{N}(\mathcal{I_\mathcal{M}})$, and vice versa. The proof follows, then, directly from Theorem \ref{th:equivalence2}.
\end{proof}

Figure \ref{fig:Pappus2} shows a sub-network of $\mathcal{N}(\mathcal{I}_{\mathcal{M}np})$ resulting from the construction of  Definition \ref{def:matnet} applied to the non-Pappus matroid $\mathcal{M}_{np}$  of Figure \ref{fig:nonPappusMat}. Node $n_1$ represents the clients in the set $R_3$ of $\mathcal{I}_{\mathcal{M}np}$, $n_2$ the basis $\{1,2,4\}$ of the non-Pappus matroid, and $n_3,n_4,n_5$ the cycle $\{1,2,3\}$.

\subsection{Special case}

Determining the capacity and design of efficient solutions for the general network coding problem is a long-standing open problem. In particular, it is currently not known whether the general problem is solvable in polynomial time, NP-hard, or even undecidable \cite{April05}. Reference \cite{LL04} proved that determining the scalar liner capacity is NP-hard, however, it is not known whether this result holds fro the vector-linear or general network codes. A important problem in this context is whether the hardness of the general network coding problem carries over to special cases of practical interest. For example,  references \cite{MEHK03,DZ06} shat that restricting the network coding problem to multiple unicast connections  does not result in a loss of generality.  In particular, \cite{DZ06} presents a procedure that transform an instance of the network coding problem into an equivalent instance with multiple unicast connections.

In this section we describe a new class of the network coding problems which captures many important properties of the general problem. The instances that belong to this class have a simple structure with a 4-partite underlying communication graph.  In particular, the class is defined by
$$\{\mathcal{N}(\mathcal{I}_{\mathcal{N}'}); \mathcal{N}' \text{ is a communication network}\},$$
where $\mathcal{I}_{\mathcal{N}'}$ is the instance of the Index Coding problem constructed as per Definition~\ref{def:1}, and $\mathcal{N}(\mathcal{I}_{\mathcal{N}'})$  is the instance of the network coding problem obtained through the construction depicted on Figure~\ref{fig:indexnet}. Figure \ref{fig:ICN}(b) shows an example of an instance of a coding network that belongs to this class. Theorem \ref{th:equivalence} implies that for any coding network $\mathcal{N}'$ it holds that $\mathcal{N}'$ has an $(n,q)$ linear network code if and only if such a solution exists for $\mathcal{N}^*=\mathcal{N}(\mathcal{I}_{\mathcal{N}'})$.

%
%
%
%

%

\begin{figure}[t]
\begin{center}
	\includegraphics[width=3in]{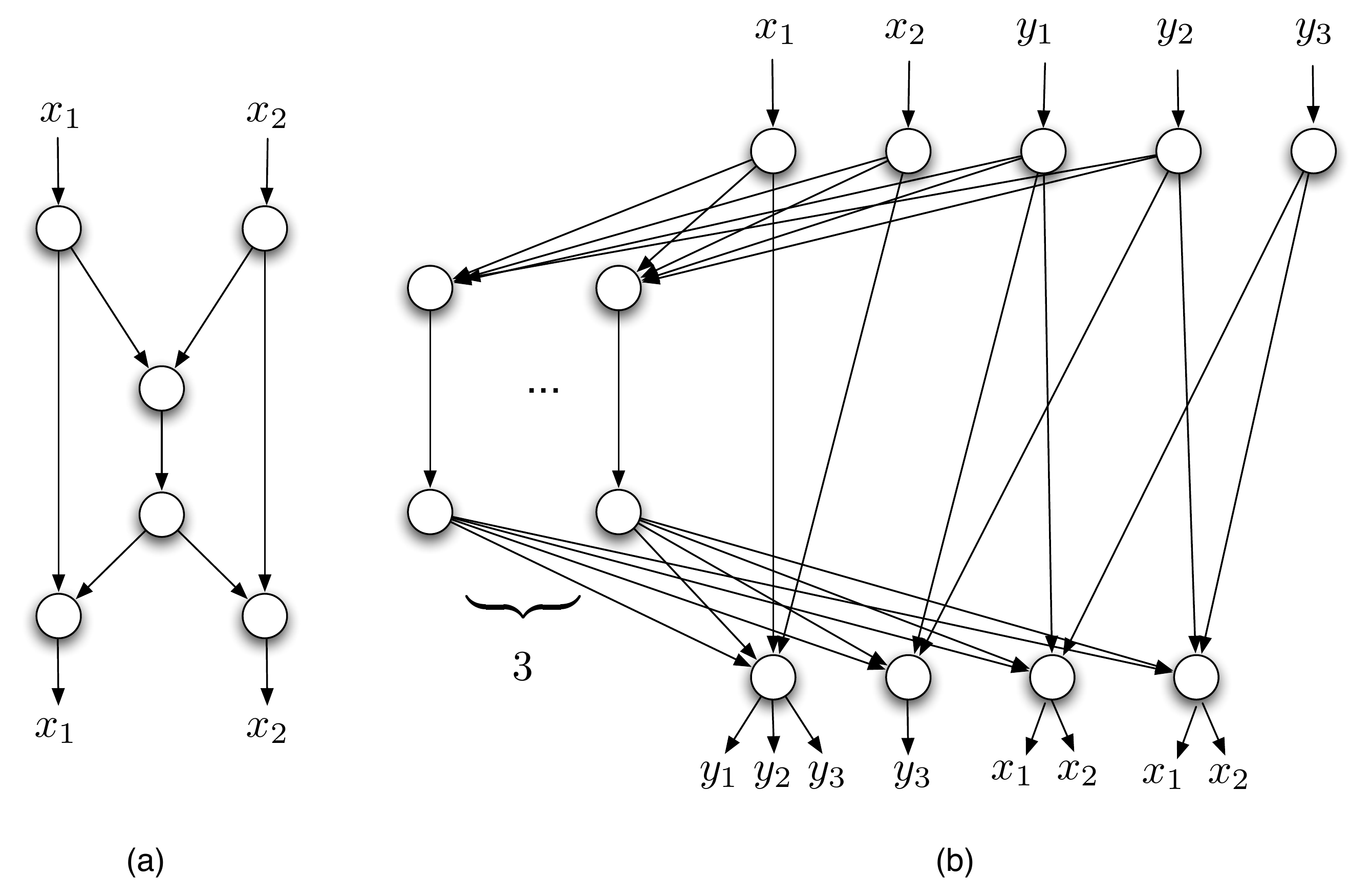}
 \caption{(a) The butterfly network. (b) the equivalent ICN network. A corresponding network code consists of the encoding functions $y_1+x_1, y_2+x_2, y_3+x_1+x_2$ on the bottleneck edge. }\label{fig:ICN}
\end{center}
\end{figure}


%


\section{Conclusion}\label{sec:conclusion}

This paper focuses on the index coding problem and its relation to the network coding problem and  matroid theory. First, we presented a reduction that maps an instance  $\mathcal{N}$ of the network coding problem to an instance $\mathcal{I}_\mathcal{N}$ of the index coding problem  such that $\mathcal{N}$ has a vector linear solution if and only if there is a perfect index code for $\mathcal{I}_\mathcal{N}$. Our reduction implies that many important results on the network coding problem carry over to the index coding problem. In particular, using the $M$-network described in \cite{MEHK03}, we show that vector linear codes outperform scalar ones. In addition, by using the results of Dougherty et al. in \cite{DFZ05} we show that non-linear codes outperform vector linear codes.

Next, we present a reduction that maps an instance of the matroid representation problem to the instance of the Index coding problem. In particular, for any given matroid
$\mathcal{M}$ we construct an instance of the index coding  problem $\mathcal{I}_\mathcal{M}$, such that $\mathcal{M}$ has a multilinear representation if and only if   $\mathcal{I}_\mathcal{M}$  has a vector linear solution over the same field. Using the properties of the non-Pappus matroid, we gave another example where vector  linear codes outperform than scalar linear ones.

%

Our results imply that there exists a strong connection between network coding and matroids. In addition, our results yield a family of coding networks that have a simple structure, but still capture many important properties of the general network coding problem.

%


\bibliographystyle{unsrt}

\bibliography{coding}

\end{document}